\documentclass[a4paper,11pt]{article}
\usepackage[utf8]{inputenc}
\usepackage[T1]{fontenc}
\usepackage{paralist}
\usepackage[vlined,ruled,linesnumbered]{algorithm2e}
\usepackage[pdftex]{graphicx}
\usepackage{amsmath}
\usepackage{amssymb}
\usepackage{amsthm}
\usepackage[sc]{mathpazo}
\usepackage{amsfonts}
\usepackage{booktabs}
\usepackage{array}
\usepackage{url}
\usepackage{hyperref}
\usepackage{mathrsfs}
\usepackage[margin=10pt, font=small, labelfont=bf]{caption}
\usepackage[margin=1in]{geometry}

\DeclareMathOperator{\E}{\mathbb{E}}
\DeclareMathOperator{\is}{IS}
\DeclareMathOperator{\es}{EM}
\DeclareMathOperator{\Pb}{Pr}

\DeclareMathOperator{\mf}{\enspace .}
\DeclareMathOperator{\mc}{\enspace ,}

\newcommand{\deq}{\triangleq}
\newcommand{\reals}{\mathbb{R}}
\newcommand{\ints}{\mathbb{Z}}

\newtheorem{theorem}{Theorem}
\newtheorem{corollary}{Corollary}
\newtheorem{proposition}{Proposition}
\newtheorem{lemma}{Lemma}

\newtheorem{example}{Example}

\theoremstyle{definition}

\newtheorem{remark}{Remark}

\theoremstyle{remark}

\title{Processor Allocation for Optimistic Parallelization of Irregular
  Programs\footnote{The original publication is available at
    www.springerlink.com \cite{VersaciP12}}}

\author{Francesco Versaci\thanks{Contact author. Email:
    \texttt{versaci@par.tuwien.ac.at}. This works was partially supported by
    PAT-INFN Project \emph{AuroraScience}, by MIUR-PRIN Project \emph{AlgoDEEP},
    and by the University of Padova Projects \emph{STPD08JA32} and
    \emph{CPDA099949}} \\{\small TU Wien \& University of Padova} \and Keshav
  Pingali\thanks{\texttt{pingali@cs.utexas.edu}}\\{\small University of Texas at
    Austin}}

\date{}

\begin{document}
\maketitle

\begin{abstract}
  Optimistic parallelization is a promising approach for the parallelization of
  irregular algorithms: potentially interfering tasks are launched dynamically,
  and the runtime system detects conflicts between concurrent activities,
  aborting and rolling back conflicting tasks.
  However, parallelism in irregular algorithms is very complex. In a regular
  algorithm like dense matrix multiplication, the amount of parallelism can
  usually be expressed as a function of the problem size, so it is reasonably
  straightforward to determine how many processors should be allocated to
  execute a regular algorithm of a certain size (this is called the processor
  allocation problem). In contrast, parallelism in irregular algorithms can be a
  function of input parameters, and the amount of parallelism can vary
  dramatically during the execution of the irregular algorithm. Therefore, the
  processor allocation problem for irregular algorithms is very difficult.

  In this paper, we describe the first systematic strategy for addressing this
  problem. Our approach is based on a construct called the \emph{conflict
    graph}, which (i) provides insight into the amount of parallelism that can
  be extracted from an irregular algorithm, and (ii) can be used to address the
  processor allocation problem for irregular algorithms. We show that this
  problem is related to a generalization of the \emph{unfriendly seating
    problem} and, by extending Turán's theorem, we obtain a worst-case class of
  problems for optimistic parallelization, which we use to derive a lower bound
  on the exploitable parallelism.
  Finally, using some theoretically derived properties and some experimental
  facts, we design a quick and stable control strategy for solving the processor
  allocation problem heuristically.
\end{abstract}

\vspace{5mm} \noindent{\bf Keywords}: Irregular algorithms, Optimistic
parallelization, Automatic parallelization, Amorphous data-parallelism,
Processor allocation, Unfriendly seating, Turán's theorem.

\pagebreak[4]

\section{Introduction}
\label{sec:an-intr-optim}
%\subsection{Preliminaries}
The advent of on-chip multiprocessors has made parallel programming a mainstream
concern. Unfortunately writing correct and efficient parallel programs is a
challenging task for the average programmer. Hence, in recent years, many
projects~\cite{KaleK93,FrigoLR98,AnJRSSTTAR01,Reinders07} have tried to automate
parallel programming for some classes of algorithms. Most of them focus on
\emph{regular} algorithms such as Fourier
transforms~\cite{FrigoJ05,PuschelMJ++05} and dense linear algebra
routines~\cite{scalapack97}. Automation is more difficult when the algorithms
are \emph{irregular} and use pointer-based data structures such as graphs and
sets.
One promising approach is based on the concept of \emph{amorphous data
  parallelism}~\cite{Mendez-LojoNPSHKBP10}.  Algorithms are formulated as
iterative computations on \emph{work-sets}, and each iteration is identified as
a quantum of work (task) that can potentially be executed in parallel with other
iterations. The Galois project \cite{GaloisPLDI11} has shown that algorithms
formulated in this way can be parallelized automatically using {\em optimistic
  parallelization}): iterations are executed speculatively in parallel and, when
an iteration conflicts with concurrently executing iterations, it is
rolled-back. Algorithms that have been successfully parallelized in this manner
include Survey propagation~\cite{BraunsteinMZ05}, Boruvka's
algorithm~\cite{Eppstein2000}, Delauney triangulation and
refinement~\cite{GuibasKS92}, and Agglomerative clustering~\cite{TanSK2005}.

In a regular algorithm like dense matrix multiplication, the amount of
parallelism can usually be expressed as a function of the problem size, so it is
reasonably straightforward to determine how many processors should be allocated
to execute a regular algorithm of a certain size (this is called the
\emph{processor allocation} problem). In contrast, parallelism in irregular
algorithms can be a function of input parameters, and the amount of parallelism
can vary dramatically during the execution of the irregular
algorithm~\cite{DBLP:conf/ppopp/KulkarniBIPC09}. Therefore, the processor
allocation problem for irregular algorithms is very difficult. Optimistic
parallelization complicates this problem even more: if there are too many
processors and too little parallel work, not only might some processors be idle
but speculative conflicts may actually retard the progress of even those
processors that have useful work to do, increasing both program execution time
and power consumption. \emph{This paper\footnotemark presents the first
  systematic approach to addressing the processor allocation problem for
  irregular algorithms under optimistic parallelization}, and it makes the
following contributions.

\footnotetext{A brief announcement of this work has been presented at SPAA'11
  \cite{VersaciP11}}

\begin{itemize}
\item We develop a simple graph-theoretic model for optimistic parallelization
  and use it to formulate processor allocation as an optimization problem that
  balances parallelism exploitation with minimizing speculative conflicts
  (Section~\ref{sec:model-optim-parall}).
\item We identify a worst-case class of problems for optimistic parallelization;
  to this purpose, we develop an extension of Turán's theorem~\cite{AlonS00}
  (Section~\ref{sec:expl-parall}).
\item Using these ideas, we develop an adaptive controller that dynamically
  solves the processor allocation problem for amorphous data-parallel programs,
  providing rapid response to changes in the amount of amorphous
  data-parallelism (Section~\ref{sec:contr-parall}).
\end{itemize}

\section{Modeling Optimistic Parallelization}
\label{sec:model-optim-parall}

A typical example of an algorithm that exhibits amorphous data-parallelism is
Dalauney mesh refinement, summarized as follows. A triangulation of some planar
region is given, containing some ``bad'' triangles (according to some quality
criterion). To remove them, each bad triangle is selected (in any arbitrary
order), and this triangle, together with triangles that lie in its {\em cavity},
are replaced with new triangles. The retriangulation can produce new bad
triangles, but this process can be proved to halt after a finite number of
steps. Two bad triangles can be processed in parallel, given that their cavities
do not overlap.

There are also algorithms, which exhibit amorphous data-parallelism, for which
the order of execution of the parallel tasks cannot be arbitrary, but must
satisfy some constraints (e.g., in discrete event simulations the events must
commit chronologically). We will not treat this class of problems in this work,
but we will focus only on \emph{unordered}
algorithms~\cite{DBLP:conf/ppopp/KulkarniBIPC09}.
A different context in which there is no roll-back and tasks do not conflict,
but obey some precedence relations, is treated in \cite{AgrawalEtal08}.

Optimistic parallelization deals with amorphous data-parallelism by maintaining
a work-set of the tasks to be executed. At each temporal step some tasks are
selected and speculatively launched in parallel. If, at runtime, two processes
modify the same data a conflict is detected and one of the two has to abort and
roll-back its execution.
Neglecting the details of the various amorphous data-parallel algorithms, we can
model their common behavior at a higher level with a simple graph-theoretic
model: we can think a scheduler as working on a dynamic graph $G_t=(V_t,E_t)$, where
the nodes represent computations we want to do, but we have no initial knowledge
of the edges, which represent conflicts between computations (see
Fig.~\ref{fig:optimistic-parallel}). At time step $t$ the system picks uniformly
at random $m_t$ nodes (the \emph{active} nodes) and tries to process them
concurrently. When it processes a node it figures out if it has some connections
with other executed nodes and, if a neighbor node happens to have been processed
before it, aborts, otherwise the node is considered processed, is removed from
the graph and some operations may be performed in the neighborhood, such as
adding new nodes with edges or altering the neighbors. The time taken to process
conflicting and non-conflicting nodes is assumed to be the same, as it happens,
e.g., for Dalauney mesh refinement.

\begin{figure}%[p]
  \centering
  \begin{tabular}{cp{5mm}cp{5mm}c}
    \includegraphics[width=35mm]{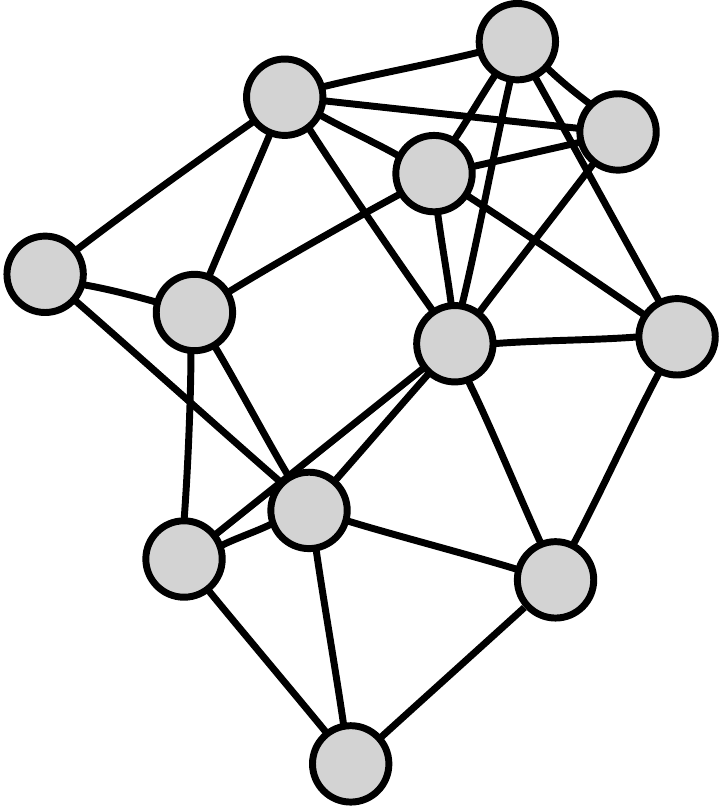} & &
    \includegraphics[width=35mm]{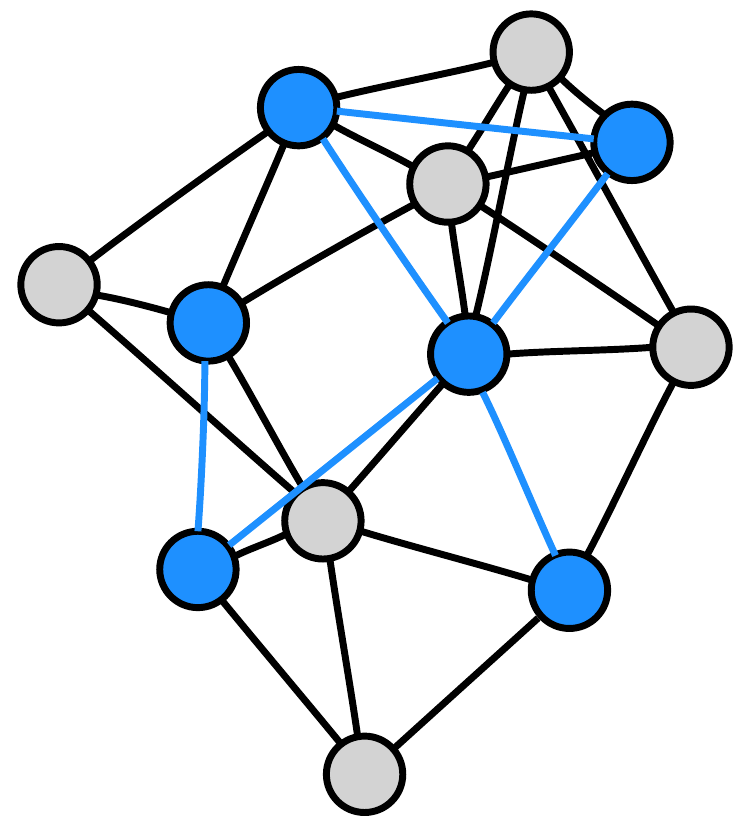} & &
    \includegraphics[width=35mm]{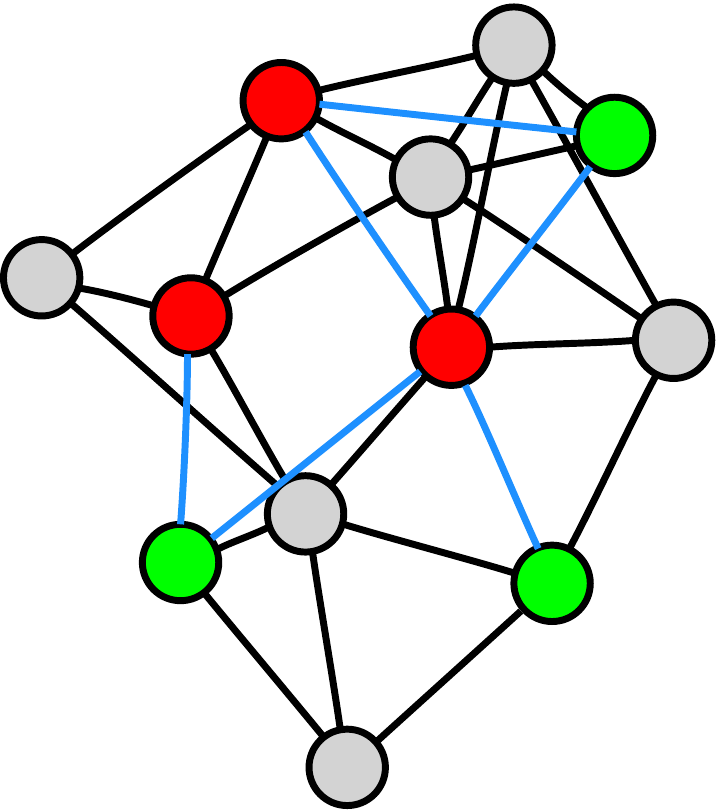} \\[2mm]
    (i) & & (ii) & & (iii)
  \end{tabular}
  \caption{Optimistic parallelization. (i) Nodes represent possible
    computations, edges conflicts between them. (ii) $m$ nodes are chosen at
    random and run concurrently. (iii) At runtime the conflicts are detected,
    some nodes abort and their execution is rolled back, leaving a maximal
    independent set in the subgraph induced by the initial nodes choice.  }
  \label{fig:optimistic-parallel}
\end{figure}

\subsection{Control Optimization Goal}
\label{sec:contr-optim-goal}

When we run an optimistic parallelization we have two contrasting goals: we both
want to maximize the work done, achieving high parallelism, but at the same time
we want to minimize the conflicts, hence obtaining a good use of the processors
time. (Furthermore, for some algorithms the roll-back work can be quite
resource-consuming.)  These two goals are not compatible, in fact if we naïvely
try to minimize the total execution time the system is forced to use always all
the available processors, whereas if we try to minimize the time wasted from
aborted processes the system uses only one processor. Therefore in the following
we choose a trade-off goal and cast it in our graph-theoretic model.

Let $G=(V,E)$ be a computations/conflicts (CC) graph with $n=|V|$ nodes. When a
scheduler chooses, uniformly at random, $m$ nodes to be run, the ordered set
$\pi_m(\cdot)$ by which they commit can be modeled as a random permutation: if
$i<j$ then $\pi_m(i)$ commits before $\pi_m(j)$ (if there is a conflict between
$\pi_m(i)$ and $\pi_m(j)$ then $\pi_m(i)$ commits and $\pi_m(j)$ aborts, if
$\pi_m(i)$ aborted due to conflicts with previous processes $\pi_m(j)$ can
commit, if not conflicting with other committed processes).  Let $k_t(\pi_m)$ be
the number of aborted processes due to conflicts and $r_t(\pi_m)\in \left[0,
  1\right)$ the ratio of conflicting processors observed at time $t$
(i.e. $r_t(\pi_m)\deq k_t(\pi_m)/m$).  We define the \emph{conflict ratio} $\bar
r_t(m)$ to be the expected $r$ that we obtain when the system is run with $m$
processors:
\begin{equation}
  \bar r_t(m)\deq \E_{\pi_m}\left[r_t(\pi_m)\right] \mc
\end{equation}
where the expectation is computed uniformly over the possible prefixes of length
$m$ of the $n$ nodes permutations.  The control problem we want to solve is the
following: \emph{given $r(\tau)$ and $m_\tau$ for $\tau<t$, choose $m_t=\mu_t$
  such that $\bar r_t(\mu_t)\simeq\rho$, where $\rho$ is a suitable parameter}.
\begin{remark}
  If we want to dynamically control the number of processors, $\rho$ must be
  chosen different from zero, otherwise the system converges to use only one
  processor, thus not being able to identify available parallelism. A value of
  $\rho \in [20\%, 30\%]$ is often reasonable, together with the constraint
  $m_t \geq 2$.
\end{remark}

\section{Exploiting Parallelism}
\label{sec:expl-parall}

In this section we study how much parallelism can be extracted from a given CC
graph and how its sparsity can affect the conflict ratio. To this purpose we
obtain a worst case class of graphs and use it to analytically derive a lower
bound for the exploitable parallelism (i.e., an upper bound for the conflict
ratio). We make extensive use of finite differences (i.e., discrete
derivatives), which are defined recursively as follows. Let $f: \ints
\rightarrow \reals$ be a real function defined on the integers, then the $i$-th
(forward) finite difference of $f$ is
\begin{align}
  \Delta^i_f(k) & = \Delta^{i-1}_f(k+1) - \Delta^{i-1}_f(k) \mc & \text{with }
  \Delta^0_f(k)&=f(k) \mf
\end{align}
(In the following we will omit $\Delta$'s superscript when equal to one, i.e.,
$\Delta\deq\Delta^1$.)

First, we obtain two basic properties of $\bar r$, which are given by the
following propositions.
\begin{proposition}\label{thm:r-increasing}
  The conflict ratio function $\bar r(m)$ is non-decreasing in $m$.
  % and its finite difference $\Delta_{\bar r}$ is non-increasing.
\end{proposition}
To prove Prop.~\ref{thm:r-increasing} we first need a lemma:
\begin{lemma}\label{lem:k-convex}
  Let $\bar k(m) \deq \E_{\pi_m} \left[ k(\pi_m) \right]$. Then $\bar k$ is a
  non-decreasing convex function, i.e.\ $\Delta_{\bar k}(m)\geq 0$ and
  $\Delta^2_{\bar k}(m)\geq 0$.
\end{lemma}
\begin{proof}%[Proof of Thm.~\ref{lem:k-convex}]
  Let $\tilde{k}(\pi_m, i)$ be the expected number of conflicting nodes running
  $r=m+i$ nodes concurrently, the first $m$ of which are $\pi_m$ and the last
  $i$ are chosen uniformly at random among the remaining ones. By definition, we
  have
  \begin{equation}
    \E_{\pi_m}\left[\tilde{k}(\pi_m, i)\right] = \bar k(m+i) \mf
  \end{equation}
  In particular,
  \begin{equation}
    \tilde{k}(\pi_m, 1)= k(\pi_m) + \Pb\left[ (m+1)\text{-th conflicts} \right] \mc
  \end{equation}
  which brings
  \begin{equation}
    \bar k(m+1)=\E_{\pi_m}\left[\tilde{k}(\pi_m, 1)\right] = \bar k(m) + \eta \mc
  \end{equation}
  with $\eta= \bar k(m+1) - \bar k(m) = \Delta_{\bar k}(m) \geq 0$, hence
  proving the monotonicity of $\bar k$. Consider now
  \begin{equation}
    \tilde{k}(\pi_m, 2)= k(\pi_m) + \Pb\left[ (m+1)\text{-th conflicts}
    \right] + \Pb\left[ (m+2)\text{-th conflicts} \right] \mf
  \end{equation}
  If the $(m+1)$-th node does not add any edge, then we have
  \begin{equation}
    \Pb\left[(m+1)\text{-th conflicts} \right]=\Pb\left[ (m+2)\text{-th
        conflicts} \right] \mc
  \end{equation}
  but since it may add some edges the probability of conflicting the second time
  is in general larger and thus $\Delta^2_{\bar k}(m)\geq 0$.
\qed\end{proof}
\begin{proof}[Prop.~\ref{thm:r-increasing}]
  Since $\bar r(m)=\bar k(m)/m$, its finite difference can be written as
  \begin{equation}
    \Delta_{\bar r}(m)=\frac{m\Delta_{\bar k}(m)-\bar k(m)}{m(m+1)} \mf
  \end{equation}
  Because of Lemma~\ref{lem:k-convex} and being $\bar k(1)=0$ we have
  \begin{equation}
    \bar k(m+1) \leq m\Delta_{\bar k}(m) \mc
  \end{equation}
  which finally brings
  \begin{equation}
    \Delta_{\bar r}(m)=\frac{m\Delta_{\bar k}(m)-\bar k(m)}{m(m+1)} \geq
    \frac{\bar k(m+1)-\bar k(m)}{m(m+1)} =
    \frac{\Delta_{\bar k}(m)}{m(m+1)} \geq 0 \mf
  \end{equation}
\qed\end{proof}

\begin{proposition}\label{prop:r-init-derivat}
  Let $G$ be a CC graph, with $n$ nodes and average degree $d$, then the initial
  derivative of $\bar r$ depends only on $n$ and $d$ as
  \begin{equation}
    \Delta_{\bar r}(1)=\frac{d}{2(n-1)} \mf
  \end{equation}
\end{proposition}
\begin{proof}
  Since
  \begin{equation}
    \Delta_{\bar r}(1) = \frac{\Delta_{\bar k}(1)-\bar k(1)}{2}=\frac{\bar
      k(2)}{2} \mc
  \end{equation}
  we just need to obtain $\bar k(2)$. Let $\tilde k$ be defined as in the proof
  on Lemma~\ref{lem:k-convex} and $\pi_1=v$ a node chosen uniformly at
  random. Then
  \begin{equation}
    \bar k(2)=\E_v \left[\tilde k(v,1)\right]=\E_v\left[\frac{d_v}{n-1}\right]
    = \frac{\E_v\left[d_v\right]}{n-1}=\frac{d}{n-1} \mf
  \end{equation}
\qed\end{proof}

A measure of the available parallelism for a given CC graph has been identified
in~\cite{KulkarniBCP09} considering, at each temporal step, a maximal
independent set of the CC graph.
The expected size of a maximal independent set gives a reasonable and computable
estimate of the available parallelism.  However, this is not enough to predict the
actual amount of parallelism that a scheduler can exploit while keeping a low conflict
ratio, as shown in the following example.
\begin{example}Let $G=K_{n^2}\cup D_n$ where $K_{n^2}$ is the complete graph of
  size $n^2$ and $D_n$ a disconnected graph of size $n$ (i.e. $G$ is made up of
  a clique of size $n^2$ and $n$ disconnected nodes). For this graph every
  maximal independent set is maximum too and has size $n+1$, but if we choose
  $n+1$ nodes uniformly at random and then compute the conflicts we obtain that,
  on average, there are only 2 independent nodes.
\end{example}

A more realistic estimate of the performance of a scheduler can be obtained by
analyzing the CC graph sparsity.  The average degree of the CC graph is linked
to the expected size of a maximal independent set of the graph by the following
well known theorem (in the variant shown in~\cite{AlonS00} or~\cite{Tao06}):
\begin{theorem}{(Turán, strong formulation).}
  Let $G=(V,E)$ be a graph, $n=|V|$ and let $d$ be the average degree of
  $G$. Then the expected size of a maximal independent set, obtained choosing
  greedily the nodes from a random permutation, is at least $s=n/(d+1)$.
\end{theorem}
\begin{remark}
  The previous bound is existentially tight: let $K^n_d$ be the graph made up of
  $s=n/(d+1)$ cliques of size $d+1$, then the average degree is $d$ and the size
  of every maximal (and maximum) independent set is exactly $s$. Furthermore,
  every other graph with the same number of nodes and edges has a bigger average
  maximal independent set.
\end{remark}

The study of the expected size of a maximal independent set in a given graph is
also known as the \emph{unfriendly seating
  problem}~\cite{FreedmanS62,FriedmanR64} and is particularly relevant in
statistical physics, where it is usually studied on mesh-like graphs
\cite{GeorgiouKK09}. The properties of the graph $K^n_d$ has suggested us the
formulation of an extension of the Turán's theorem. We prove that the graphs
$K^n_d$ provide a worst case (for a given degree $d$) for the generalization of
this problem obtained by focusing on maximal independent set of induced
subgraphs.  This allows, when given a target conflict ratio $\rho$, the
computation of a lower bound for the parallelism a scheduler can exploit.
\begin{theorem}\label{thm:tur-extension}
  Let $G$ be a graph with same nodes number and degree of $K^n_d$ and let
  $\es_m(G)$ be the expected size of a maximal independent set of the subgraph
  induced by a uniformly random choice of $m$ nodes in $G$, then
  \begin{equation}
    \es_m(G) \geq \es_m(K^n_d) \mf
  \end{equation}
\end{theorem}
To prove it we first need the following lemma.
\begin{lemma}\label{lem:eta-convexity}
  The function $ \eta_j(x)\deq\prod_{i=1}^j (n-i-x) $ is convex for $x\in[0, n-j]$.
\end{lemma}
\begin{proof}
  We prove by induction on $j$ that, for $x\in[0, n-j]$,
  \begin{align}
    \eta_j(x) &\geq 0 \mc & \eta'_j(x) &\leq 0 \mc & \eta''_j(x) &\geq 0  \mf
  \end{align}
  \paragraph{Base case} Let $\eta_0(x)=1$. The properties above are easily verified.
  \paragraph{Induction} Since $\eta_{j}(x)=\eta_{j-1}(x) (n-j-x)$, we obtain
  \begin{equation}
    \eta'_j(x) = -\eta_{j-1}(x) + (n-j-x)\eta'_{j-1}(x) \mc
  \end{equation}
  which is non-positive by inductive hypotheses. Similarly,
  \begin{equation}
    \eta''_j(x) = -2\eta'_{j-1}(x) + (n-j-x)\eta''_{j-1}(x)
  \end{equation}
  is non-negative.
\qed\end{proof}
\begin{proof}[Thm.~\ref{thm:tur-extension}]
  Consider a random permutation $\pi$ of the nodes of a generic graph $G$ that
  has the same number of nodes and edges of $K^n_d$. We assume the prefix of
  length $m$ of $\pi$ (i.e.\ $\pi(1), \ldots, \pi(m)$) forms the active nodes
  and focus on the following independent set $\is_m$ in the subgraph induced: a
  node $v$ is in $\is_m(G,\pi)$ if and only if it is in the first $m$ positions
  of $\pi$ and it has no neighbors preceding it. Let $b_m(G)$ be the expected size
  of $\is_m(G,\pi)$ averaged over all possible $\pi$'s (chosen uniformly):
  \begin{equation}
    b_m(G) \deq \E_\pi\left[\# \is_m(G,\pi)\right] \mf
  \end{equation}
  Since for construction $b_m(G)\leq\es_m(G)$ whereas
  $b_m(K^n_d)=\es_m(K^n_d)$, we just need to prove that $b_m(K^n_d)\leq b_m(G)$.
  Given a generic node $v$ of degree $d_v$ and a random permutation $\pi$, its
  probability to be in $\is_m(G,\pi)$ is
  \begin{equation}
    \Pb\left[v \in \is_m(G,\pi) \right] = \frac{1}{n} \sum_{j=1}^m \prod_{i=1}^{j-1}
    \frac{n-i-d_v}{n-i} \mf
  \end{equation}
  By the linearity of the expectation we can write $b$ as
  \begin{align}
    b_m(G) &= \frac{1}{n} \sum_{v=v_1}^{v_n} \sum_{j=1}^m \prod_{i=1}^{j-1} \frac{n-i-d_v}{n-i}
    = \E_v\left[\sum_{j=1}^m \prod_{i=1}^{j-1} \frac{n-i-d_v}{n-i}\right] \mc \\
    b_m(K^n_d) &= \sum_{j=1}^m \prod_{i=1}^{j-1} \frac{n-i-d}{n-i}
    = \sum_{j=1}^m \prod_{i=1}^{j-1} \frac{n-i-\E_v[d_v]}{n-i} \mf
  \end{align}
  To prove that $\es_m(G)\geq \es_m(K^n_d)$ is thus enough showing that
  \begin{equation}
    \forall j \quad
    \E_v\left[ \prod_{i=1}^{j} (n-i-d_v) \right] \geq  \prod_{i=1}^{j}
    \left(n-i-\E_v[d_v]\right) \mc
  \end{equation}
  which can be done applying Jensen's inequality~\cite{Jensen1906}, since in
  Lemma~\ref{lem:eta-convexity} we have proved the convexity of
  $\eta_{j}(x)\deq\prod_{i=1}^{j} (n-i-x)$.
\qed\end{proof}
\begin{corollary}
  \label{cor:op-worst}
  The worst case for a scheduler among the graphs with the same number of nodes and
  edges is obtained for the graph $K^n_d$ (for which we can analytically
  approximate the performance, as shown in \S\ref{sec:analysys-worst-case}).
\end{corollary}
\begin{proof}
  Since
  \begin{equation}
    \bar r(m) =\frac{m-\es_m(G)}{m} = 1 - \frac{1}{m} \es_m(G) \mc
  \end{equation}
  the thesis follows.
\qed\end{proof}
\subsection{Analysis of the Worst-Case Performance}
\label{sec:analysys-worst-case}
\begin{theorem}
  \label{thm:worst-case}
  Let $d$ be the average degree of $G=(V,E)$ with $n=|V|$ (for simplicity we
  assume $n/(d+1)\in\mathbb{N}$). The conflict ratio is bounded from above as
  \begin{equation}
    \bar r(m) \leq 1-\frac{n}{m(d+1)} \left( 1 - \prod_{i=1}^m \frac{n-d-i}{n+1-i}
    \right) \mf
  \end{equation}
\end{theorem}
\begin{proof}
  Let $s=n/(d+1)$ be the number of connected components in $K_d^n$.  Because of
  Thm.~\ref{thm:tur-extension} and Cor.~\ref{cor:op-worst} it suffices to show
  that
  \begin{equation}
    \es_m(K_d^n) = s \left( 1 - \prod_{i=1}^m \frac{n-d-i}{n+1-i} \right) \mf
  \end{equation}
  The probability for a connected component $k$ of $K_d^n$ not to be accessed when
  $m$ nodes are chosen is given by the following hypergeometric
  \begin{equation}
    \Pr[\text{$k$ not hit}] = \frac{\begin{pmatrix}  n-d-1 \\
        m  \end{pmatrix}\begin{pmatrix}  d+1 \\
        0  \end{pmatrix}}{\begin{pmatrix}  n \\ m  \end{pmatrix}}
    = \prod_{i=1}^m \frac{n-d-i}{n+1-i} \mf
  \end{equation}
  Let $X_k$ be a random variable that is 1 when component $k$ is hit and 0
  otherwise. We have that $\E[X_k]=1-\prod_{i=1}^m \frac{n-d-i}{n+1-i}$ and, by
  the linearity of the expectation, the average number of components accessed is
  \begin{equation}
    \E\left[\sum_{k=1}^s X_k\right]=\sum_{k=1}^s \E[X_k] = s \left( 1 - \prod_{i=1}^m
      \frac{n-d-i}{n+1-i} \right) \mf
  \end{equation}
\qed\end{proof}
\begin{corollary}
  \label{cor:approx-bound}
  When $n$ and $m$ increase the bound is well approximated by
  \begin{equation}
    \bar r(m) \leq 1 - \frac{n}{m(d+1)} \left[ 1 - \left(1 -
        \frac{m}{n}\right)^{d+1}\right] \mf
  \end{equation}
\end{corollary}
\begin{proof}
  Stirling approximation for the binomial, followed by low order terms deletion
  in the resulting formula.
\qed\end{proof}
\begin{corollary}
  \label{cor:approx-bound-2}
  If we set $m=\alpha s= \frac{\alpha n}{d+1}$ we obtain
  \begin{align}
    \bar r(m) & \leq 1 - \frac{1}{\alpha} \left[ 1 - \left(1 -
        \frac{\alpha}{d+1}\right)^{d+1}\right]  \leq 1-\frac{1}{\alpha}\left[1 -
      e^{-\alpha}\right] \mf
  \end{align}
\end{corollary}

\section{Controlling Processors Allocation}
\label{sec:contr-parall}

In this section we will design an efficient control heuristic that dynamically
chooses the number of processes to be run by a scheduler, in order to obtain high
parallelism while keeping the conflict ratio low.
In the following we suppose that the properties of $G_t$ are varying slowly
compared to the convergence of $m_t$ toward $\mu_t$ under the algorithm we will
develop (see \S\ref{sec:algorithm-evaluation}), so we can consider $G_t=G$ and
$\mu_t=\mu$ and thus our goal is making $m_t$ converge to $\mu$.

Since the conflict ratio is a non-decreasing function of the number of launched
tasks $m$ (Prop.~\ref{thm:r-increasing}) we could find $m\simeq\mu$ by bisection
simply noticing that
\begin{align}
   & \bar r(m') \leq  \rho \leq \bar r(m'')  & \Rightarrow&   & m' \leq \mu \leq m'' \mf
\end{align}
The control we propose is slightly more complex and is based on recurrence
relations, i.e., we compute $m_{t+1}$ as a function $F$ of the target conflict
ratio $\rho$ and of the parameters which characterize the system at the previous
timestep:
\begin{equation}
  m^F_{t+1}=F\left(\rho,r_t,m_t\right) \mf
\end{equation}
The initial value $m_0$ for a recurrence can be chosen to be 2 but, if we have
an estimate of the CC graph average degree $d$, we can choose a smarter value:
in fact applying Cor.~\ref{cor:approx-bound-2} we are sure that using, e.g.,
$m=\frac{n}{2(d+1)}$ processors we will have at most a conflict ratio of
$21.3\%$.

\begin{algorithm}
  \small
  \tcp{Tunable parameters}
  $m_0 = 2;\qquad$
  $m_{\max} = 1024;\qquad$
  $m_{\min} = 2$\;

  $T = 4;\qquad$
  $r_{\min} = 3\%;\qquad$
  $\alpha_0 = 25\%;\qquad$
  $\alpha_1 = 6\%$\;

  \tcp{Variables}
  $m \leftarrow m_0;\qquad$
  $r \leftarrow 0;\qquad$
  $t \leftarrow 0$\;

  \tcp{Main loop}
  \While{nodes to elaborate $\not = 0$}{
    $t \leftarrow t+1$\;

    \lIf{$m>m_{\max}$}{$m\leftarrow m_{\max}$\;}
    \lElseIf{$m<m_{\min}$}{$m\leftarrow m_{\min}$\;}

    \textrm{Launch the scheduler with $m$ nodes}\;
    $r \leftarrow  r + \textrm{new conflict ratio}$\;

    \If{$(t \bmod T)=T-1$}{
      $r \leftarrow r/T$\;
      $\displaystyle\alpha \leftarrow \left| 1-\frac{r}{\rho}\right|$\;

      \uIf{$\alpha>\alpha_0$}{
        \lIf{$r<r_{\min}$}{$r\leftarrow r_{\min}$\;}
        $\displaystyle m \leftarrow \left\lceil \frac{\rho}{r} \, m \right\rceil $\;
      }
      \uElseIf{$\alpha>\alpha_1$}{
        $\displaystyle m \leftarrow \left\lceil ( 1-r+\rho) \, m \right\rceil $\;
      }
      $r \leftarrow 0;$

    }
  }
  \caption{Pseudo-code of the proposed hybrid control algorithm}
  \label{hybcode}
\end{algorithm}

Our control heuristic (Algorithm~\ref{hybcode}) is a hybridization of two simple
recurrences.  The first recurrence is quite natural and increases $m$ based on
the distance between $r$ and $\rho$:
\begin{align}
  \text{\bf Recurrence A:} \qquad m^A_{t+1}=(1-r_t+\rho) m_t \mf
\end{align}
The second recurrence exploits some experimental facts.  In Fig.~\ref{barr} we
have plotted the conflict ratio functions for three CC graphs with the same size
and average degree (note that initial derivative is the same for all the graphs,
in accordance with Prop.~\ref{prop:r-init-derivat}).  We see that conflict
ratios which reach a high value ($\bar r(n)>\frac{1}{2}$) are initially well
approximated by a straight line (for $m$ such that $\bar r(m)\leq \rho = 20\div
30 \%$), whereas functions that deviates from this behavior do not raise too
much. This suggests us to assume an initial linearity in controlling $m_t$, as
done by the following recurrence:
\begin{align}
  \text{\bf Recurrence B:} \qquad m^B_{t+1}=\frac{\rho}{r_t} m_t \mf
\end{align}

\begin{figure}%[p]
  \centering
  \includegraphics[width=10cm]{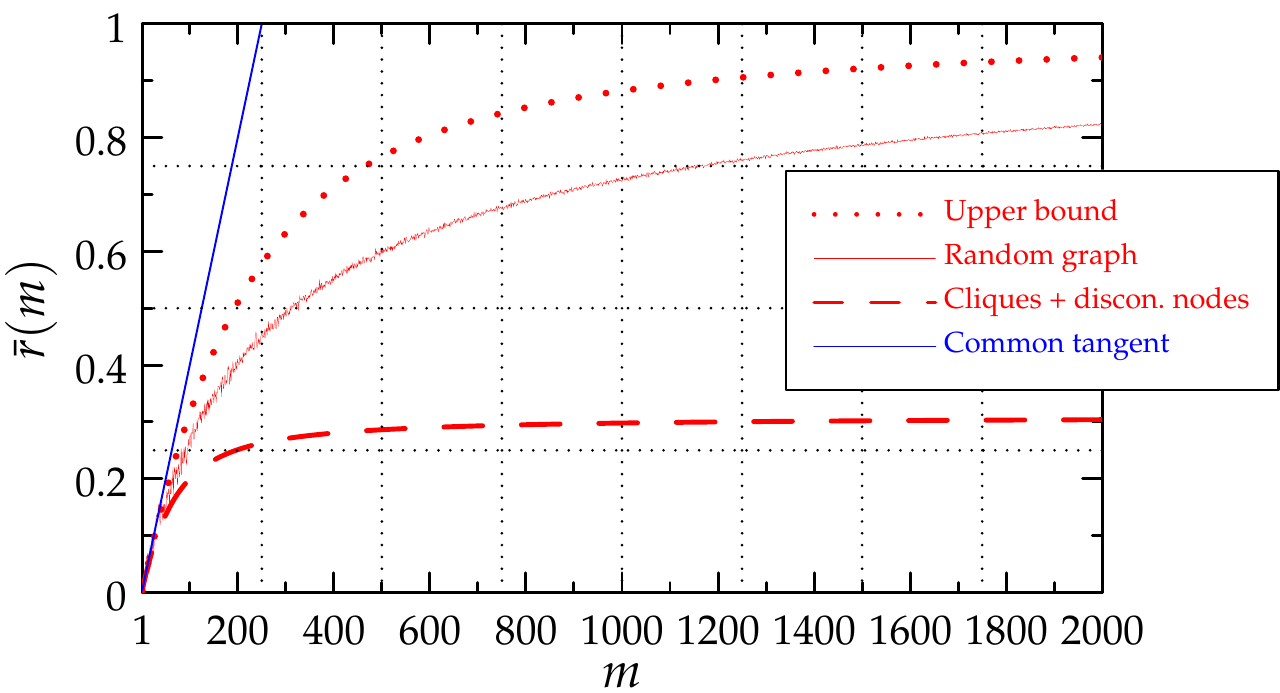}
  \caption{A plot of $\bar r(m)$ for some graphs with $n=2000$ and $d=16$: (i)
    the worst case upper bound of Cor.~\ref{cor:approx-bound} (ii) a random
    graph (edges chosen uniformly at random until desired degree is reached;
    data obtained by computer simulation) (iii) a graph unions of cliques and
    disconnected nodes.}
  \label{barr}
\end{figure}

The two recurrences can be roughly compared as follows (see Fig.~\ref{hyb}):
Recurrence A has a slower convergence than Recurrence B, but it is less
susceptible to noise (the variance that makes $r_t$ realizations different from
$\bar r_t$). This is the reason for which we chose to merge them in an hybrid
algorithm: initially, when the difference between $r$ and $\rho$ is big, we use
Recurrence~B to exploit its quick convergence and then Recurrence~A is adopted,
for a finer tuning of the control.

\begin{figure}%[p]
  \centering
  \includegraphics[width=10cm]{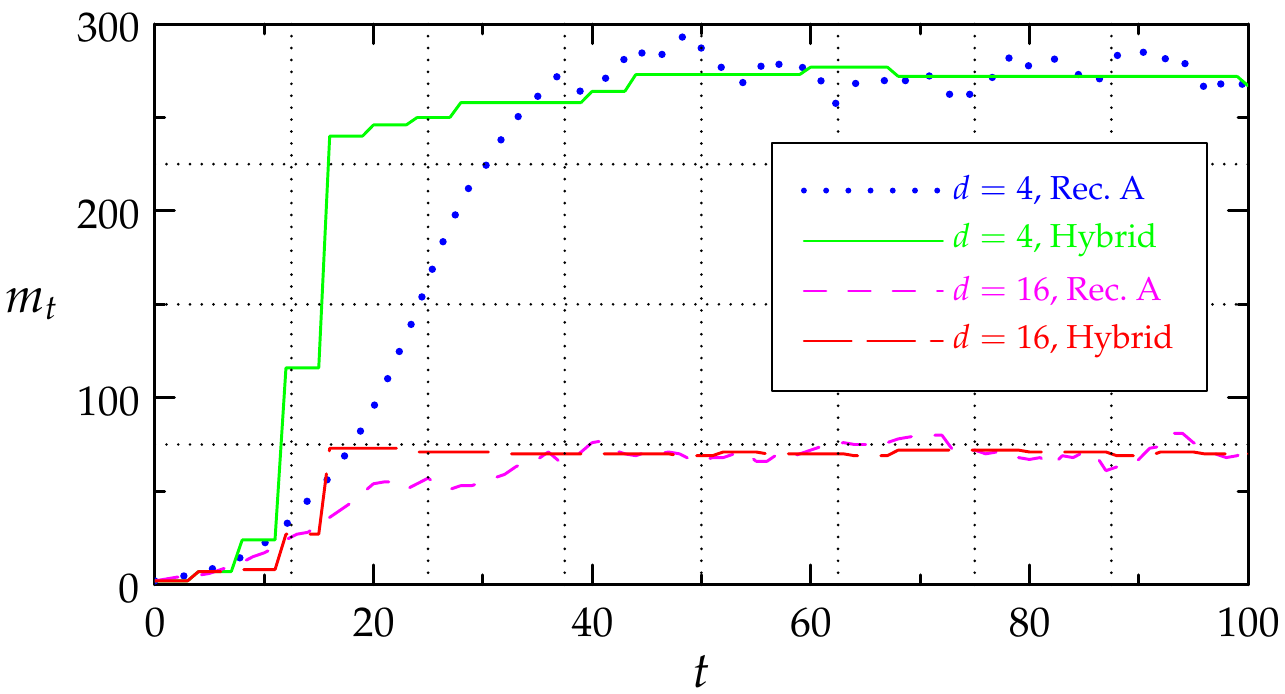}
  \caption{Comparison between two realizations of the hybrid algorithm and one
    that only uses Recurrence~A, for two different random graphs ($n=2000$ in
    both cases). The hybrid version has different parameters for $m$ greater or
    smaller than 20. $\rho$ was chosen to be $20\%$. The proposed algorithm
    proves to be both quick in convergence and stable.}
  \label{hyb}
\end{figure}

\subsection{Experimental Evaluation}
\label{sec:algorithm-evaluation}

In the practical implementation of the control algorithm we have made the
following optimizations:
\begin{compactitem}
\item Since $r_t$ can have a big variance, especially when $m$ is small, we
  decided to apply the changes to $m$ every $T$ steps, using the averaged values
  obtained in these intervals, to smooth the oscillations.
\item To further reduce the oscillations we apply a change only if the observed
  $r_t$ is sufficiently different from $\rho$ (e.g. more than $6\%$), thus
  avoiding small variations in the steady state, which interfere with locality
  exploitation because of the data moving from one processor to another.
\item Another problem that must be considered is that for small values of $m$
  the variance is much bigger, so it is better to tune separately this case
  using different parameters (this optimization is not shown in the
  pseudo-code).
\end{compactitem}

To validate our controller we have run the following simulation: a random CC
graph of fixed average degree $d$ is taken and the controller runs on it,
starting with $m_0=2$. We are interested in seeing how many temporal steps it
takes to converge to $m_t\simeq \mu$.
As can be seen in~\cite{KulkarniBCP09} the parallelism profile of many practical
applications can vary quite abruptly, e.g., Delauney mesh refinement can go from
no parallelism to one thousand possible parallel tasks in just 30 temporal
steps. Therefore, an algorithm that wants to efficiently control the processors
allocations for these problems must adapt very quickly to changes in the
available parallelism. Our controller, that uses the very fast Recurrence~B in
the initial phase, proves to do a fast enough job: as shown in Fig.~\ref{hyb}
in about 15 steps the controller converges close to the desired $\mu$ value.

\section{Conclusions and Future Work}
\label{sec:concl-future-work}

Automatic parallelization of irregular algorithms is a rich and complex subject
and will offer many difficult challenges to researchers in the next future.
In this paper we have focused on the processor allocation problem for unordered
data-amorphous algorithms; it would be extremely valuable to obtain similar
results for the more general and difficult case of \emph{ordered} algorithms
(e.g., discrete event simulation), in particular it is very hard to obtain good
estimates of the available parallelism for such algorithms, given the complex
dependencies arising between the concurrent tasks. Another aspect which needs
investigation, especially in the ordered context, is whether some statical
properties of the behavior of irregular algorithms can be modeled, extracted and
exploited to build better controllers, able to dynamically adapt to the
different execution phases.

As for a real-world implementation, the proposed control heuristic is now being
integrated in the Galois system and it will be evaluated on more realistic
workloads.
%%
%% (Given the nature of the processor allocation problem the controller will
%% show gain full performance once the system will be able to manage a larger
%% number of processors.)

\section*{Acknowledgments}
\label{sec:acknowledgments}

We express our gratitude to Gianfranco Bilardi for the valuable feedback on
recurrence-based controllers and to all the Galois project members for the
useful discussions on optimistic parallelization modeling.

% \pagebreak[4]
% \nocite{*}

\bibliographystyle{splncs03}
\bibliography{procontrol}

\end{document}